\documentclass[copyright,creativecommons,sharealike]{eptcs}
\providecommand{\event}{CoALP-Ty 2016} 
\usepackage{underscore}           
\usepackage{ amsthm }
\usepackage{ enumerate }
\usepackage{tabularx}
\usepackage{ amsmath }
\usepackage{todonotes}
\usepackage{color}

\title{Structural Resolution with Co-inductive Loop Detection}
\author{Yue Li
\institute{School of Mathematical \& Computer Sciences\\
Heriot-Watt University\\
Edinburgh, United Kingdom}
\email{yl55@hw.ac.uk}}
\def\titlerunning{Co-inductive S-resolution}
\def\authorrunning{Yue Li}

\theoremstyle{definition}
\newtheorem{defn}{Definition}
\newtheorem{exmp}{Example}[defn]

\theoremstyle{plain}
\newtheorem{thm}{Theorem}
\newtheorem{prop}{Proposition}
\newtheorem{lem}{Lemma}

\theoremstyle{remark}
\newtheorem*{rmk}{Remark}
\begin{document}
\maketitle
\begin{abstract}
A way to combine co-SLD style loop detection with structural resolution was found and is introduced in this work, to extend structural resolution with co-induction. In particular, we present the operational semantics, called co-inductive structural resolution, of this novel combination and prove its soundness with respect to the greatest complete Herbrand model.

\end{abstract}
\section{Introduction}
Co-inductive logic programming extends traditional logic programming by enabling co-inductive reasoning to deal with infinite SLD-derivation, which has practical implication in different fields of computing such as model checking, planning as well as type inference~\cite{SimonCoSLD2006,SimonThesis2006,AbstractCompAnconaCLD10,CoALPSemanImpKKJPMS16}.

One operational semantics of co-inductive logic programming is co-inductive SLD resolution (co-SLD)~\cite{SimonCoSLD2006,SimonThesis2006}, which combines loop detection with traditional SLD resolution, so that it can be used to reason co-inductively about infinite rational terms.  An alternative operational semantics for co-inductive logic programming is co-algebraic logic programming ~\cite{CoALPSemanImpKKJPMS16,StrResoKomendantskayaJ15}, which adopts a more general co-induction rule and uses structural resolution instead of SLD resolution as its induction rule. 

The distinctive features of co-algebraic logic programming, compared with co-SLD, include that the co-inductive reasoning mechanism of the former goes beyond loop detection; as a result the computed formulae by the former are allowed to be either rational terms or (finite observation of) irrational terms. Moreover,  structural resolution \cite{StrResoKomendantskayaJ15,RewTreeJohannKK15} allows for analysis of productivity \cite{StrResoKomendantskayaJ15} of logic programs. Productivity, also known as computations at infinity \cite[ch. 4]{FoundOfLPLloyd1987}, concerns computation of infinite data structures by non-terminating derivations. Productivity is studied not only in logic programming community, but also in functional programming community, where they developed decision algorithm for co-recursive list and stream definitions \cite{SijtsmaProductivityOfRecursiveListDefn,Endrullis2010PoductiveStream}.   Those distinctive features make co-algebraic logic programming an ideal basis for developing productivity decision algorithms.

In this paper we explore a combination of co-SLD style loop detection with structural resolution, resulting in a novel operational semantics called \emph{co-inductive structural resolution} (co-S-resolution for short), and we prove its soundness with respect to the greatest complete Herbrand model. An implementation is also presented as a contribution. 

 Co-inductive structural resolution is created as an intermediate semantics between co-SLD and co-algebraic logic programming. On the one hand, it inherits loop detection from co-SLD, which is a simpler co-inductive reasoning mechanism compared with the co-inductive mechanism of co-algebraic logic programming. On the other hand, it uses structural resolution as co-algebraic logic programming does, allowing for analysis of productivity. Introducing such an intermediate semantics has its practical implication: there are two challenges involved in the design of productivity decision algorithms, which are 1) productivity analysis with structural resolution and 2) co-inductive reasoning beyond loop detection; introducing the intermediate semantics makes it possible to deal with these \emph{two} challenges \emph{one at a time}, so co-S-resolution prepares future development of a productivity decision algorithm that uses loop detection, which in turn will prepare even further algorithm development that goes beyond loop detection. These points will be further discussed in Section~\ref{sec:related work}.

 An overview of the rest of the paper is as follows. In Section~\ref{sec: prelim} we will introduce preliminary concepts that cover substitution and unification with rational trees, co-SLD and structural resolution, and greatest fix point, which will prepare us for further theory construction in later sections. In Section~\ref{sec: co-SR} we will introduce the semantics for co-inductive structural resolution and prove its soundness.  In Section~\ref{sec:related work} we will have a review of related work, discuss the importance of co-S-resolution for productivity decision, and conclude the paper.  Appendix~\ref{app: imple} presents the implementation.  

\section{Preliminaries}\label{sec: prelim}

We assume readers' understanding of standard definition of \emph{first order term} and the modelling of (possibly infinite) terms by \emph{trees}. ``Term'' and ``tree'' are used interchangeably in this paper. Details about these concepts can be found in~\cite{FoundOfLPLloyd1987,TreeByCOURCELLE1983}. 

\begin{defn}[Rational Term]
A \emph{rational term} \cite{PrologIn10Figures,JaffarInfTreeLP,AnconaCoSLD15,SimonCoSLD2006} refers to a (possibly non-ground and possibly infinite) term (or tree) that has a  \emph{finite} amount of distinct sub-terms (or sub-trees). A rational term is also known as a \emph{regular term} \cite{TreeByCOURCELLE1983,Gupta2007}.     

\end{defn}

Our definition for \emph{substitution with rational trees} (referred to as \emph{substitution} for short hereinafter) inherits the principle of substitution with finite trees in logic programming \cite{FoundOfLPLloyd1987}.

\begin{defn}[Substitution]
A substitution is a mapping of the form $S=\{x_1/t_1,\ldots,x_n/t_n\}$ where $x_1,\ldots,x_n$ are distinct variables, $t_1,\ldots,t_n$ are rational terms, and $\forall\ i,j\in\{1,\ldots,n\}$, $x_i$ does not occur in $t_j$. Moreover, $\epsilon$ denotes the empty substitution.
\end{defn}
\begin{exmp}
Let $\theta=\{x_1/f(f(\ldots)), x_2/g(x_3)\}$ be a substitution. Applying $\theta$ to the term $p(x_1,x_2)$ is denoted by $p(x_1,x_2)\theta$, which evaluates to $p(f(f(\ldots)),g(x_3))$.
\end{exmp}

Composition of substitutions is defined in the same way as in~\cite[Sec. 4]{FoundOfLPLloyd1987}. 

\begin{defn}[Unification]\label{defn: unification}
Given two rational terms $t_1$ and $t_2$, unification is the process of finding a substitution (unifier) $\theta$ such that $t_1\theta=t_2\theta$,  i.e. applying $\theta$ separately to $t_1$ and $t_2$ yields the same tree. This relation is denoted by $t_1\sim_\theta t_2$.
\end{defn}
  The standard approach to rational term unification \cite{MartelliUnification,TreeByCOURCELLE1983,PrologIn10Figures,ALogicalReconstructionOfPrologII,SemanOfPrologWithoutOccursCheckWeijland1988} involves \emph{systems of equations of finite terms}, and \emph{transforms} that turn equation systems to their \emph{reduced form} as the output of the unification algorithm. The reduced form equation system can further be \emph{solved} to obtain a solution in the domain of rational trees \cite{TreeByCOURCELLE1983}.  We omit the details of the unification algorithm for rational trees and the details of solving equation systems, which can be found in the above literature. 
 
\begin{rmk}
Our definition of substitution refers to solutions of reduced form equation systems.
Consider the unification problem $p(X)\sim p(f(X))$. The standard approach regards $p(X)\sim p(f(X))$ as an equation system $\{p(X)=p(f(X))\}$ and reduces it to the reduced form $\{X= f(X)\}$, which is called a substitution in the standard sense. Locally in this paper we \emph{solve} the reduced form $\{X= f(X)\}$ to obtain the solution $\{X=f(f(\ldots))\}$ and call this solution a substitution. We believe that our treatment of substitution can emphasize the variables that are to be instantiated and will make the theory about co-inductive structure resolution easier to formulate and understand.       
\end{rmk} 
  \emph{Term matching} is a concept closely related to unification, and is prerequisite for rewriting reduction in structural resolution \cite{StrResoKomendantskayaJ15,RewTreeJohannKK15}. We extend applicable terms for matching from finite trees to rational trees by building the concept of rational tree term matching on the concept of rational tree unification. 

\begin{defn}[Term Matching]\label{defn: term matching}
Given two rational terms $t_1$, $t_2$ and a unifier $\sigma$, if $\sigma$ also satisfies that $t_1\sigma=t_2$, then it is said that $t_1$ \emph{subsumes} $t_2$, or $t_1$ \emph{matches against} $t_2$. This relation is denoted by $t_1\prec_\sigma t_2$. $\sigma$ is called a \emph{matcher} for $t_1$ against $t_2$. 
\end{defn}

\begin{rmk}[Uniqueness of matcher]
If two terms have matchers $\sigma_1$ and $\sigma_2$, then $\sigma_1$ equals to $\sigma_2$ when restricted to variables that occur in the terms. Nevertheless, a matcher $\sigma$ for term $t_1$ against $t_2$ is intended to be identity on variables not in $t_1$.
\end{rmk}

The symbolic notation for unification ($\sim$) and matching ($\prec$) follows~\cite{StrResoKomendantskayaJ15}. Note that ($\sim$) is a symmetric relation but ($\prec$) is not symmetric. Sometimes $t_1\prec t_2$ may fail to convey which term subsumes (i.e. matches against) which. A mnemonic tip is to regard the ``precede'' symbol ($\prec$) as the ``less than'' symbol ($<$) and derive from $t_1\prec t_2$ that $t_2$ might be bigger (i.e. contains more symbols) than $t_1$.

\begin{exmp}
\begin{enumerate}[a)]
\item$p(x)\sim_\theta p(f(x))$ where $\theta=\{x/f(f(f(\ldots)))\}$. $\theta$ is not a matcher.

\item $p(x_1,x_1)\sim_\theta p(f(y_1),y_1)$ where $\theta=\{x_1/f(f(f(\ldots))), y_1/f(f(f(\ldots)))\}$. $\theta$ is not a matcher.

\item $p(x_1,x_2)\sim_\theta p(f(y_1),y_1)$ where $\theta=\{x_1/f(y_1),x_2/y_1\}$. $\theta$ is a matcher and $p(x_1,x_2)\prec_\theta p(f(y_1),y_1)$. 
\end{enumerate}

\end{exmp}

  We now introduce the operational semantics of the well-known SLD-resolution ($\mathbf{L}$inear resolution for $\mathbf{D}$efinite clauses with $\mathbf{S}$election function) \cite{FoundOfLPLloyd1987,KowalVESemanLP76} as a precursor of co-SLD and of structural resolution.

\begin{defn}[SLD-resolution]
Given a logic program P and goal \[G={\leftarrow A_1,\ldots,A_n}\] if there exists in P a clause $B_0\leftarrow B_1,\ldots,B_m$ (with freshly renamed variables), such that $B_0\sim_\theta A_k$ for some $k\in\{1,\ldots,n\} $, then by SLD-resolution we derive \[G'={\leftarrow (A_1,\ldots,A_{k-1},B_1,\ldots,B_m,A_{k+1},\ldots,A_n)\theta}\] 
\end{defn}
\begin{rmk}
The notation of the form $(A_1,\ldots,A_n)\theta$ denotes application of $\theta$ to every $A_i,\ i\in\{1,\ldots,n\}$.
\end{rmk}
 In the following definition of co-SLD we introduce a set $S$ for each predicate $A$ in a goal \cite{AnconaCoSLD15}, where $S$ records all previous goals (or their instances) that are relevant to the co-inductive proof of $A$.  

\begin{defn}[co-SLD resolution]\label{defn: co-SLD reso}
Given a logic program P and a goal \[G={\leftarrow (A_1,S_1),\ldots,(A_n, S_n)}\] the next goal $G'$ can be derived by one of the following two rules:
\begin{enumerate}
\item  If there exists in P a clause $B_0\leftarrow B_1,\ldots,B_m$ (with freshly renamed variables), such that $B_0\sim_\theta A_k$ for some $k\in\{1,\ldots,n\}$, then let $S'=S_k\cup \{A_k\}$, we derive \[G'={\leftarrow \big((A_1, S_1),\ldots,(A_{k-1}, S_{k-1}),(B_1, S'),\ldots,(B_m, S'),(A_{k+1}, S_{k+1}),\ldots,(A_n,S_n)\big)\theta}\] 

\item (Loop Detection) If $A_k\sim_\theta B$ for some $k\in\{1,\ldots,n\}$ and some $B\in S_k$, we derive \[G'={\leftarrow \big((A_1, S_1),\ldots,(A_{k-1}, S_{k-1}),(A_{k+1}, S_{k+1}),\ldots,(A_n,S_n)\big)\theta}\]
\end{enumerate}
\end{defn}

\begin{rmk}
The notation of the form $\big((A_1,S_1),\ldots,(A_n, S_n)\big)\theta$ denotes application of $\theta$ to every $A_i$ and to every member of every $S_i$, $i\in\{1,\ldots,n\}$. 
\end{rmk}

\begin{defn}[co-SLD Derivation/Refutation]
A co-SLD derivation consists of a possibly infinite sequence of goals $G_0,G_1,\ldots$ where $G_0$ is of the form $\leftarrow (A_1,\emptyset),\ldots,(A_m, \emptyset)\ (m\geq 0)$, and for all $i\geq 0$, $G_{i+1}$ is derived from $G_i$ using co-SLD resolution. A finite co-SLD derivation ending with the empty goal is called an co-SLD refutation\footnote{Logic programming works by refuting the goal, which is the negation of the proposition that is to be proven. }.  
\end{defn} 

\begin{defn}[Computed Substitution]
Given a co-SLD refutation $\mathcal{D}$, let $\theta_1,\theta_2,\ldots,\theta_n$ be the sequence of unifiers computed in $\mathcal{D}$ in the same order as they were computed, their composition $\theta_1\theta_2\cdots\theta_n$ is called the computed substitution from $\mathcal{D}$. 
\end{defn}

Co-SLD is co-inductively sound \cite{AnconaCoSLD15,SimonCoSLD2006}. Co-inductive soundness is defined in terms of the \emph{greatest complete Herbrand model}. We assume the standard definition of complete Herbrand interpretation and complete Herbrand base \cite[Sec. 25]{FoundOfLPLloyd1987}, which, compared with Herbrand interpretation/base, allow for infinite ground terms and atoms in addition to finite ones.

\begin{defn}[$T'_P$ operator]
Let P be a logic program and $B'_P$ be P's complete Herbrand base. The complete immediate consequence operator $T'_P:\ 2^{B'_P}\mapsto2^{B'_P}$ is defined as follows. Let $I\subseteq B'_P$ be a complete Herbrand interpretation. Then \[T'_P(I)=\{A\in B'_P\mid A\gets A_1,\ldots,A_n\text{ is a ground instance of a clause in P and }\{A_1,\ldots,A_n\}\subseteq I \}\]
\end{defn}

\begin{defn}[Greatest complete Herbrand model]
Let P be a program. The greatest fix point $\text{gfp}(T'_P)=\cup\{I\mid I \subseteq T'_P(I)\}$ of $T'_P$ is called the greatest complete Herbrand model of P.
\end{defn}

More details on $T'_P$ operator and its fix points can be found in e.g.~\cite[Sec. 26]{FoundOfLPLloyd1987}. We now justify the loop detection rule of co-SLD with an example.

\begin{exmp}
Consider the following program P which defines co-recursively all streams of 0's and 1's.
\begin{flushleft}
\hspace{1em}bit(0)$\leftarrow$

\hspace{1em}bit(1)$\leftarrow$

\hspace{1em}bit-stream(cons(X,Xs)) $\gets$ bit(X), bit-stream(Xs)
\end{flushleft}
The co-SLD refutation for goal $\gets$bit-stream(cons(0,Xs)), following the left-first computation rule, is finished by one step of loop detection which unifies bit-stream(Xs) and bit-stream(cons(0,Xs)) with unifier $\theta$=\{Xs/cons(0,cons(0,\ldots))\}. Note that 
\begin{flushleft}
\hspace{1em}bit-stream(cons(0,Xs)) $\gets$ bit(0), bit-stream(Xs) \hfill (*)
\end{flushleft}
is an instance of the program clause, to which we can apply unifier $\theta$ and we get another program clause instance   
\begin{flushleft}
\hspace{1em}bit-stream(cons(0,cons(0,\ldots))) $\gets$ bit(0), bit-stream(cons(0,cons(0,\ldots))) \hfill (**)
\end{flushleft}
We regard (*) and (**) as proof trees \cite[Sec. 1.6]{clark1980predicate}, and notice that applying the loop detection rule extends (*) into (**), whose set $I$ of nodes satisfies $I \subseteq T'_P(I)$, therefore $I$ is a subset of gfp($T'_P$). The establishment of the relation $I \subseteq T'_P(I)$ is mainly due to the reasoning that for each atom $A\in I$ (or, for each node $A$ of proof tree (**)), there exists a program clause instance whose head is $A$, and whose body is a subset of $I$, so that if $A$ is in $I$, then $A$ is in $T'_P(I)$, indicating $I \subseteq T'_P(I)$. Such reasoning applies to all co-SLD refutation that involves use of loop detection.
\end{exmp}

\begin{defn}[Structural Resolution]
Given a logic program P and goal \[G={\leftarrow A_1,\ldots,A_n}\] the next goal $G'$ is derived using one of the following two rules:

\begin{enumerate}
\item (Rewriting Reduction) If there exists in P a clause $B_0\leftarrow B_1,\ldots,B_m$ (with freshly renamed variables), such that $B_0\prec_\theta A_k$ for some $k\in\{1,\ldots,n\}$, then we derive \[G'={\leftarrow (A_1,\ldots,A_{k-1},B_1,\ldots,B_m,A_{k+1},\ldots,A_n)\theta}\]
\item (Substitution Reduction) If there exists in P a clause $B_0\leftarrow B_1,\ldots,B_m$ (with freshly renamed variables), such that $B_0\sim_\theta A_k$ but not $B_0\prec_\theta A_k$ for some $k\in\{1,\ldots,n\}$, then we derive \[G'={\leftarrow (A_1,\ldots,A_n)\theta}\]
\end{enumerate} 
\end{defn}

\begin{rmk}
About rewriting reduction, notice that it is a special case of SLD-resolution, and since matcher $\theta$ only instantiates variables from the renamed clause $B_0\leftarrow B_1,\ldots,B_m$ without instantiating variables from goal $G$, the derived goal $G'$ can also be written as $G'={\leftarrow A_1,\ldots,A_{k-1},(B_1,\ldots,B_m)\theta,A_{k+1},\ldots,A_n}$.  About substitution reduction, notice that it is, by nature, instantiation of universal quantifier. 
\end{rmk}
\begin{defn}[S-Derivation/Refutation]
 A structural resolution derivation (S-derivation for short) consists of a possibly infinite sequence $G_0,G_1,\ldots$ of goals such that for all $i\geq 0$, $G_{i+1}$ is derived from $G_i$ using structural resolution, without consecutive use of substitution reduction. A finite S-derivation ending with the empty goal is called an S-refutation. 
\end{defn}
 \begin{defn}[Computed Substitution]
 Given a S-refutation $\mathcal{D}$, let $\theta_1,\theta_2,\ldots,\theta_n$ be the sequence of unifiers computed in $\mathcal{D}$ due to application of substitution reduction, sorted in the same order as they were computed, their composition $\theta_1\theta_2\cdots\theta_n$ is called the computed substitution from $\mathcal{D}$. 
 \end{defn}

\begin{exmp}\label{exmp: struc reso}
Consider the program:
\begin{center} 
\begin{tabularx}{\textwidth}{ XXX }
\hspace{1em}  $p(f(X))\leftarrow q(X)$ & $q(a)\leftarrow$ & $r(f(a))\leftarrow$
\end{tabularx}
\end{center}
Given goal $\leftarrow p(X),r(X)$ the next goal is $\leftarrow p(f(X_1)),r(f(X_1))$ by substitution reduction on $p(X)$ (with renamed clause $p(f(X_1))\leftarrow q(X_1)$ and unifier $\theta_1=\{X/f(X_1)\}$), then the next goal is $\leftarrow q(X_1),r(f(X_1))$ by rewriting reduction on $p(f(X_1))$ (with renamed clause $p(f(X_2))\leftarrow q(X_2)$ and matcher $\{X_2/X_1\}$), then the next goal is $\leftarrow q(a),r(f(a))$ by substitution reduction on $q(X_1)$ (with unifier $\theta_2=\{X_1/a\}$). Two more steps of rewriting reduction derive the empty goal $\leftarrow$ which terminates successfully the resolution and the computed substitution is the composition $\theta_1\theta_2=\{X/f(a),X_1/a\}$. 
\end{exmp}

For more details on structural resolution, see \cite{StrResoKomendantskayaJ15,RewTreeJohannKK15,CoALPSemanImpKKJPMS16}. Next we introduce the combination of structural resolution and co-SLD style loop detection.

\section[Co-inductive Semantics]{Co-inductive Structural Resolution}\label{sec: co-SR}
We introduce the declarative and operational semantics of co-inductive structural resolution. For the operational semantics we introduce how it was formulated and prove its co-inductive soundness.

\subsection{Declarative Semantics}
The declarative semantics of co-inductive structural resolution is chosen to be the greatest fixed point over the complete Herbrand base \cite[ch. 4]{FoundOfLPLloyd1987}\cite{SemanOfPrologWithoutOccursCheckWeijland1988}, as for co-SLD \cite{SimonCoSLD2006,SimonThesis2006,AnconaCoSLD15}. In fact, it was a conjecture \cite{KatyaCommunication} that some form of combination of structural resolution and loop detection is correct w.r.t.\ the greatest complete Herbrand model as co-SLD is, since they share the same co-induction mechanism and their inductive components (structural resolution and SLD resolution, respectively) are both sound and complete w.r.t.\ the least Herbrand model \cite{StrResoKomendantskayaJ15}.
\subsection{Operational Semantics}   
The implementation presented in Appendix~\ref{app: imple} played important role in formulation of the operational semantics. The implementation was created by integrating existing implementation of structural resolution \cite{Yue17SresoImple} and co-SLD \cite{Ancona13coSLDinProlog}, which showed plausible behaviour. So the implementation was then abstracted to obtain the operational semantics, whose soundness was later proved. The upshot is that the implementation had come before the formulation of the operational semantics, but was then verified as the soundness of the operational semantics was proved. In this section we present the operational semantics.

\begin{defn}[Co-inductive Structural Resolution]\label{defn: co-s-reso}
Given a logic program P and goal \[G={\leftarrow (A_1,S_1),\ldots,(A_n, S_n)}\] the next goal $G'$ can be derived by one of the following three rules:
\begin{enumerate}
\item  (Rewriting Reduction) If there exists in P a clause $B_0\leftarrow B_1,\ldots,B_m$ (with freshly renamed variables), such that $B_0\prec_\theta A_k$ for some $k\in\{1,\ldots,n\}$, then let $S'=S_k\cup \{A_k\}$, we derive \[G'={\leftarrow (A_1, S_1),\ldots,(A_{k-1}, S_{k-1}),(B_1\theta, S'),\ldots,(B_m\theta, S'),(A_{k+1}, S_{k+1}),\ldots,(A_n,S_n)}\] 
\item (Substitution Reduction) If there exists in P a clause $B_0\leftarrow B_1,\ldots,B_m$ (with freshly renamed variables), such that $B_0\sim_\theta A_k$ but not $B_0\prec_\theta A_k$ for some $k\in\{1,\ldots,n\}$, then we derive \[G'={\leftarrow \big((A_1, S_1),\ldots,(A_n,S_n)\big)\theta}\]
\item (Loop Detection) If $A_k\sim_\theta B$ for some $k\in\{1,\ldots,n\}$ and some $B\in S_k$, we derive \[G'={\leftarrow \big((A_1, S_1),\ldots,(A_{k-1}, S_{k-1}),(A_{k+1}, S_{k+1}),\ldots,(A_n,S_n)\big)\theta}\] 
\end{enumerate}
\end{defn}

Notice that the Loop Detection rule for co-inductive structural resolution is the same as its counterpart in co-SLD, and rule-1 of co-inductive structural resolution is a special case of rule-1 of co-SLD. 
\begin{defn}[co-S-Derivation/Refutation]\label{defn: co-s-deriv}
A co-inductive structural resolution derivation is a possibly infinite sequence $G_0,G_1,\ldots$ where $G_0$ is of the form $\leftarrow(A_1, \emptyset),\ldots,(A_m,\emptyset)\ (m\geq 0)$, and for all $i\geq 0$, $G_{i+1}$ is derived from $G_i$ by co-S-resolution without consecutive use of substitution reduction. A finite co-S-derivation ending with the empty goal is called a co-S-refutation.
 
\end{defn}
\begin{defn}[Computed Substitution]
Given a co-S-refutation $\mathcal{D}$, let $\theta_1,\theta_2,\ldots,\theta_n$ be the sequence of unifiers computed in $\mathcal{D}$ due to application of rule-2 or rule-3, sorted in the same order as they were computed, their composition $\theta_1\theta_2\cdots\theta_n$ is called the computed substitution from $\mathcal{D}$. 
\end{defn}

\begin{exmp}\label{exmp: compreh co-s-reso}
Consider program:
\begin{center}
\begin{tabularx}{\textwidth}{ XXX }
\hspace{1em}$p(s(X)) \leftarrow q(X)$ & $ q(X)\leftarrow p(X),r(X)$ & $r(X) \leftarrow$ \\
\end{tabularx}
\end{center}
In the following co-S-refutation, for each goal we always select the left most predicate to resolve.
\vspace{1em}
\begin{enumerate}[{\textrm{Goal} 1:}]
\item $\leftarrow\big(q(X),\emptyset\big)$
\item $\leftarrow\big(p(X),\{q(X)\}\big),\big(r(X),\{q(X)\}\big)$ \hfill (rule-1.~$ q(X_1)\leftarrow p(X_1),r(X_1).~\{X_1/X\}$)
\item $\leftarrow\big(p(s(X_2)),\{q(s(X_2))\}\big),\big(r(s(X_2)),\{q(s(X_2))\}\big)$ \hfill (rule-2.~$ p(s(X_2))\leftarrow q(X_2).~\theta_1=\{X/s(X_2)\}$)
\item $\leftarrow\big(q(X_2),\{p(s(X_2)),q(s(X_2))\}\big),\big(r(s(X_2)),\{q(s(X_2))\}\big)$ \hfill (rule-1.~$ p(s(X_3))\leftarrow q(X_3).~\{X_3/X_2\}$)
\item $\leftarrow\big(r(s(s(s(\ldots)))),\{q(s(s(s(\ldots))))\}\big)$ \hfill (rule-3.~$ q(X_2)\sim_{\theta_2} q(s(X_2)).~\theta_2=\{X_2/s(s(s(\ldots)))\}$)
\item $\leftarrow$ \hfill (rule-1.~$ r(X_4)\leftarrow.~\{X_4/s(s(s(\ldots)))\}$)
\end{enumerate}\vspace{1em}
The answer to Goal 1 is given by computed substitution $\theta_1\theta_2=\{X/s(s(s(\ldots))),X_2/s(s(s(\ldots)))\}$. Loop detection is used once for reduction from Goal 4 to Goal 5, and predicate $q(X_2)$ in Goal 4 is co-inductively proved.    
\end{exmp}

\subsection{Soundness Proof}

In this section, the main result shows that given a goal $G$, if there is a co-S-refutation for $G$, with computed substitution $\sigma$, then there is a co-SLD refutation for $G\sigma$, with computed substitution $\epsilon$ (i.e. the empty substitution). Therefore if co-SLD is sound, then $G\sigma\epsilon=G\sigma$ is in the greatest complete Herbrand model, meaning that co-S-resolution is also sound.

We will define a transformation algorithm that step-by-step transforms a co-S-refutation of goal $G$ into a co-SLD refutation of goal $G\sigma$. The transformation is via an intermediate derivation, called \emph{co-rewriting-id derivation}, which simply consists of co-SLD resolution steps interleaved with identity reduction steps (c.f. Definition~\ref{defn: id reduction}).  So given a co-S-refutation, it will be firstly transformed into a co-rewriting-id refutation, which will then be trivially transformed into a co-SLD refutation. The transformation from a co-S-refutation to a co-rewriting-id refutation is done during a sequential traverse of the co-S-refutation, starting from the initial goal. According to the three lemmas (i.e. Lemma~\ref{lem: rew preserv},~\ref{lem: instant preserv} and~\ref{lem: loop detect preserve}, defined later), each goal reduction step in the co-S-refutation establishes a co-rewriting-id reduction step, and all co-rewriting-id reduction steps established during the traverse form the co-rewriting-id refutation. The following are details of the proof.

\begin{defn}[Identity Reduction]\label{defn: id reduction}
Given some goal $G$, the reduction from $G$ to itself, is called identity reduction, denoted by \[G\xrightarrow{\textrm{id}}G \]
\end{defn}

\begin{defn}[co-Rewriting-ID Resolution]\label{defn: co-rew-id reso}
Given a program and some goal $G$, the next goal $G'$ can be derived from $G$ using one of following three rules:
\begin{enumerate}
\item The same as rule 1 in Definition~\ref{defn: co-s-reso}.
\item Identity reduction.
\item The same as rule 3 in Definition~\ref{defn: co-s-reso}.
\end{enumerate} 
\end{defn}

\begin{defn}[co-Rewriting-ID Derivation/Refutation]\label{defn: co-rew-id deriv}
A co-rewriting-id derivation is a possibly infinite sequence $G_0,G_1,\ldots$  where $G_0$ is of the form $\leftarrow(A_1, \emptyset),\ldots,(A_m,\emptyset)\ (m\geq 0)$, and for all $i\geq 0$, $G_{i+1}$ is derived from $G_i$ by co-rewriting-id resolution without consecutive use of identity reduction. A finite co-rewriting-id derivation ending with the empty goal is called a co-rewriting-id refutation.

\end{defn}
\begin{defn}[Computed Substitution]
Given a co-rewriting-id refutation $\mathcal{D}$, let $\theta_1,\theta_2,\ldots,\theta_n$ be the sequence of unifiers computed in $\mathcal{D}$ due to application of rule-3, sorted in the same order as they were computed, their composition $\theta_1\theta_2\cdots\theta_n$ is called the computed substitution from $\mathcal{D}$. 
\end{defn}
\begin{prop}\label{thm: co-rew-id to cosld}
Given a program, for any goal $G$, if there is a co-rewriting-id refutation for $G$ with computed substitution $\theta$, then there is a co-SLD refutation for $G$ with the same computed substitution $\theta$. 
\end{prop}
\begin{proof}
Suppose $\mathcal{D}=G_0,\ldots,G_n$ is a co-rewriting-id refutation for $G=G_0$ with computed substitution $\theta$. By simultaneously removing from $\mathcal{D}$ all $G_{i+1}\ (i\in [0,n-1])$ such that $G_i\xrightarrow{\textrm{id}} G_{i+1}$, the resulting derivation $\mathcal{D'}$ constitutes a co-SLD derivation with computed substitution $\theta$. Moreover, $\mathcal{D'}$ is a special case of co-SLD derivation since Definition~\ref{defn: co-rew-id reso}-rule 1 is a special case of Definition~\ref{defn: co-SLD reso}-rule 1.
\end{proof}

\begin{thm}\label{thm: co-s to co-r-id}
Given a program, for any goal $G$, if there is a co-S-refutation for $G$ with computed substitution $\sigma$, then there is a co-rewriting-id refutation for $G\sigma$ with computed substitution $\epsilon$ (the empty substitution).
\end{thm}

The proof of Theorem~\ref{thm: co-s to co-r-id} is based on properties of co-inductive structural resolution rules, formulated in the following three lemmas.

\begin{lem}[Rewriting Preservation]\label{lem: rew preserv}
Let \[G\xrightarrow[B]{\textrm{rule-1}}G'\] be a goal reduction using rule-1 (as defined in Definition~\ref{defn: co-s-reso}),  and $B$ the program clause involved in the reduction. 

Then for any substitution $\sigma$, it holds that \[G\sigma\xrightarrow[B]{\textrm{rule-1}}G'\sigma\] 
\end{lem} 
\begin{proof}
Assume 
\begin{itemize}
\item $G={\leftarrow(A_1,S_1),\ldots,(A_k,S_k),\ldots,(A_n,S_n)}$ and
\item $B$ has the form $B_0\leftarrow B_1,\ldots,B_m\ (m\geq 0)$ and
\item $B_0\prec_\gamma A_k$, for some $k\in[1,n]$.
\end{itemize} 
By Definition~\ref{defn: co-s-reso}, rule-1, 
\begin{equation}\label{equa: G prime in rew preserve}
G'={\leftarrow(A_1,S_1),\ldots,(A_{k-1},S_{k-1}),(B_1\gamma,S'),\ldots,(B_m\gamma,S'),(A_{k+1},S_{k+1}),\ldots,(A_n,S_n)}
\end{equation} 
where $S'=S_k\cup\{A_k\}.$

Since $B_0\prec_\gamma A_k$ (by the above assumption), it means (by Definition~\ref{defn: term matching}) that 
\begin{equation}\label{equa: B 0 gamma equals to A k in rew pres}
B_0\gamma=A_k
\end{equation}
Then for all $\sigma$, if we apply $\sigma$ to both sides of \eqref{equa: B 0 gamma equals to A k in rew pres}, we have $B_0\gamma\sigma=A_k\sigma$, which means (by associativity of substitution  \cite[Sec. 4]{FoundOfLPLloyd1987} and Definition~\ref{defn: term matching}) that
\begin{equation}\label{equa: B0 also subsumes A k sigma}
B_0\prec_{\gamma\sigma} A_k\sigma. 
\end{equation}
Now consider $G\sigma$, by notational convention,
\begin{equation}\label{equa: G sigma form in rew preserv}
 G\sigma={\leftarrow(A_1\sigma,S_1\sigma),\ldots,(A_k\sigma,S_k\sigma),\ldots(A_n\sigma,S_n\sigma)}
\end{equation} 
Because of \eqref{equa: B0 also subsumes A k sigma} and \eqref{equa: G sigma form in rew preserv}, we can have reduction  
\begin{equation}\label{equa: reduction of G sigma in rew preserv}
G\sigma\xrightarrow[B]{\textrm{rule-1}}G''
\end{equation}
where, by Definition~\ref{defn: co-s-reso}, rule-1,  
\begin{equation}\label{equa: G double prime in rew preserv}
G''={\leftarrow(A_1\sigma,S_1\sigma),\ldots,(A_{k-1}\sigma,S_{k-1}\sigma),(B_1\gamma\sigma,S''),\ldots,(B_m\gamma\sigma,S''),(A_{k+1}\sigma,S_{k+1}\sigma),\ldots,(A_n\sigma,S_n\sigma)} 
\end{equation}
where $S''=S_k\sigma \cup \{A_k\sigma\}.$

Compare \eqref{equa: G prime in rew preserve} and \eqref{equa: G double prime in rew preserv}, we have, by notational convention, 
\begin{equation}\label{equa: G double prime is G prime sigma}
G''=G'\sigma
\end{equation}
By \eqref{equa: reduction of G sigma in rew preserv} and \eqref{equa: G double prime is G prime sigma}, we reach the conclusion of Lemma~\ref{lem: rew preserv}.
\end{proof}

Hereinafter we adopt the following notation for substitution compositions. Given a sequence of substitutions $\theta_1,\theta_2,\ldots,\theta_n$, for all $k\in\{1,\ldots,n\}$, let $\sigma_k$ denote the composition $\theta_k\theta_{k+1}\cdots\theta_n$.  For example,
let $\theta_1,\theta_2,\theta_3,\theta_4$ be a sequence of 4 substitutions, then $\sigma_1=\theta_1\theta_2\theta_3\theta_4$, $\sigma_2=\theta_2\theta_3\theta_4$, $\sigma_3=\theta_3\theta_4$ and $\sigma_4=\theta_4$.

\begin{lem}[Instantiation Preservation]\label{lem: instant preserv}
Let \[G\xrightarrow[\theta_k]{\textrm{rule-2}}G'\] be a goal reduction using rule-2 (as defined in Definition~\ref{defn: co-s-reso}), where  $\theta_k$ is the unifier involved in the reduction, and let \[\sigma_k=\theta_k\theta_{k+1}\cdots\theta_n\] for some $n> k$ and some (arbitrary and possibly $\epsilon$) substitutions $\theta_{k+1},\ldots,\theta_n$.

Then \[G\sigma_k\xrightarrow{\ \textrm{id}\ }G'\sigma_{k+1}\]
\end{lem}
\begin{rmk}
The sequence of $\theta$'s in Lemma~\ref{lem: instant preserv} will come from co-S-resolution steps when Lemma~\ref{lem: instant preserv} is used to prove Theorem~\ref{thm: co-s to co-r-id}.
\end{rmk}
\begin{proof}
From the premise of Lemma~\ref{lem: instant preserv}, we have 
\begin{equation}\label{equa: G prime equals to G theta k in instan pres}
G\theta_k=G'
\end{equation}
Applying $\sigma_{k+1}(=\theta_{k+1}\cdots\theta_n)$ to both sides of \eqref{equa: G prime equals to G theta k in instan pres} we have  
\begin{equation}\label{equa: sigma k plus one applied to both sides in instan pres}
G\theta_k\sigma_{k+1}=G'\sigma_{k+1}
\end{equation}
Note that in \eqref{equa: sigma k plus one applied to both sides in instan pres}  \[\theta_k\sigma_{k+1}=\sigma_k\] therefore by associativity of substitution, \eqref{equa: sigma k plus one applied to both sides in instan pres} can be written as  \[G\sigma_k=G'\sigma_{k+1}\] hence the identity reduction $G\sigma_k\xrightarrow{\ \textrm{id}\ }G'\sigma_{k+1}$.
\end{proof}

\begin{lem}[Loop Detection Preservation]\label{lem: loop detect preserve}
Let \[G\xrightarrow[\theta_k]{\textrm{rule-3}}G'\] be a goal reduction using rule-3 (as defined in Definition~\ref{defn: co-s-reso}), where  $\theta_k$ is the unifier involved in the reduction, and let \[\sigma_k=\theta_k\theta_{k+1}\cdots\theta_n\] for some $n> k$ and some (arbitrary and possibly $\epsilon$) substitutions $\theta_{k+1},\ldots,\theta_n$.

Then \[G\sigma_k\xrightarrow[\epsilon]{ \textrm{rule-3} }G'\sigma_{k+1}\]
\end{lem}
\begin{proof}
Assume 
\begin{equation}
G={\leftarrow(A_1,S_1),\ldots,(A_k,S_k),\ldots,(A_n,S_n)}
\end{equation}
and 
\begin{equation}\label{equa: loop pres A k unifies B}
A_k\sim_{\theta_k}B
\end{equation}
for some $k\in[1,n]$ and some 
\begin{equation}\label{equa: B in S k loop pres}
B\in S_k
\end{equation}
By Definition~\ref{defn: co-s-reso}, rule-3, 
\begin{equation}\label{equa: G prime verbose}
G'={\leftarrow\big((A_1,S_1),\ldots,(A_{k-1},S_{k-1}),(A_{k+1},S_{k+1}),\ldots,(A_n,S_n)\big)\theta_k}
\end{equation} 

From \eqref{equa: loop pres A k unifies B} and Definition~\ref{defn: unification},
\begin{equation}\label{equa: A k theta k equals to b theta k}
A_k\theta_k = B\theta_k
\end{equation}
Applying $\sigma_{k+1}(=\theta_{k+1}\cdots\theta_n)$ to both sides of \eqref{equa: A k theta k equals to b theta k}, we have $A_k\theta_k\sigma_{k+1} = B\theta_k\sigma_{k+1}$, then due to $\sigma_k=\theta_k\sigma_{k+1}$ and associativity of substitution, 
\begin{equation}
A_k\sigma_{k} = B\sigma_{k}
\end{equation}
which means
\begin{equation}\label{equa: loop new unif}
A_k\sigma_{k}\sim_\epsilon B\sigma_k
\end{equation}

Consider $G\sigma_k$, which can be written as 
\begin{equation}\label{equa: loop G sigma k form}
G\sigma_k={\leftarrow(A_1\sigma_k,S_1\sigma_k),\ldots,(A_k\sigma_k,S_k\sigma_k),\ldots,(A_n\sigma_k,S_n\sigma_k)}
\end{equation}
Due to \eqref{equa: B in S k loop pres}, it holds that 
\begin{equation}\label{equa: B sigma k in S k sigma k loop}
B\sigma_k\in S_K\sigma_k
\end{equation} 
From \eqref{equa: loop new unif} and \eqref{equa: B sigma k in S k sigma k loop}, $G\sigma_k$ as in \eqref{equa: loop G sigma k form} can be reduced using Definition~\ref{defn: co-s-reso} rule 3, resulting in 
\begin{equation*}
G''={\leftarrow(A_1\sigma_k,S_1\sigma_k),\ldots,(A_{k-1}\sigma_k,S_{k-1}\sigma_k),(A_{k+1}\sigma_k,S_{k+1}\sigma_k),\ldots,(A_n\sigma_k,S_n\sigma_k)}
\end{equation*}which can be rewritten in a simpler form 
\begin{equation}\label{equa: G sigma k reduced loop}
G''={\leftarrow\big((A_1,S_1),\ldots,(A_{k-1},S_{k-1}),(A_{k+1},S_{k+1}),\ldots,(A_n,S_n)\big)\sigma_k}
\end{equation}
The reduction of $G\sigma_k$ is denoted by  
\begin{equation}\label{equa: loop pres reduction rule-3}
G\sigma_k\xrightarrow[\epsilon]{ \textrm{rule-3} }G''
\end{equation}

Comparing \eqref{equa: G prime verbose} with \eqref{equa: G sigma k reduced loop}, we conclude, by associativity of substitution, that \[G''=G'\sigma_{k+1}\] and with \eqref{equa: loop pres reduction rule-3} we reach the conclusion of Lemma~\ref{lem: loop detect preserve}.
\end{proof}
 Next we give an algorithm that outputs co-rewriting-id refutations, giving a co-S-refutation as input. This algorithm constitutes our proof of Theorem~\ref{thm: co-s to co-r-id} and we will provide an example to demonstrate the algorithm at work.  

\begin{proof}[Proof of Theorem~\ref{thm: co-s to co-r-id}]
Given a goal $G=G_0$, assume $\mathcal{D}=G_0,\ldots,G_n$ is a co-S-derivation, and $\theta_1,\ldots\theta_m$ is the sequence of unifiers computed during derivation $\mathcal{D}$ due to the use of Definition~\ref{defn: co-s-reso} rule 2 or 3. Let $\theta_{m+1}=\epsilon$ so $\sigma_1=\theta_1,\ldots\theta_m\theta_{m+1}$ is the computed substitution for goal $G$.

Definition~\ref{defn: co-s-reso} is the default domain when we mention rule 1,2 or 3 in this proof. We build the co-rewriting-id derivation $\mathcal{D'}$ of $G\sigma_1$ using the following algorithm.

For each $i\in\{0,\ldots,n-1\}$, \emph{starting from $i=0$ and in the ascending order of $i$}, \hfill (*)

\begin{itemize}
\item If \[G_i\xrightarrow[B]{\textrm{rule-1}}G_{i+1}\] then write down, by Lemma~\ref{lem: rew preserv}, \[G_i\sigma_x\xrightarrow[B]{\textrm{rule-1}}G_{i+1}\sigma_x\]
where \[ x= \begin{cases}
     1        & \quad \text{If }i=0;\\
     k        & \quad \text{If } i>0\text{ and }G_i\sigma_k\text{ is in }\mathcal{D'}.\\
  \end{cases}
\]
$x$ is well-defined in the second case because of the condition (*).
\item If \[G_i\xrightarrow[\theta_k]{\textrm{rule-2}}G_{i+1}\] then write down, by Lemma~\ref{lem: instant preserv}, \[G_i\sigma_k\xrightarrow{\textrm{id}}G_{i+1}\sigma_{k+1}\]
\item If \[G_i\xrightarrow[\theta_k]{\textrm{rule-3}}G_{i+1}\] then write down, by Lemma~\ref{lem: loop detect preserve}, \[G_i\sigma_k\xrightarrow[\epsilon]{\textrm{rule-3}}G_{i+1}\sigma_{k+1}\]
\end{itemize}  

\end{proof}
\begin{rmk}
Compare the specific co-S-refutation for goal $G_0$, which has computed substitution $\sigma_1=\theta_1\theta_2\theta_3\theta_4$ where $\theta_4=\epsilon$,  with the co-rewriting-id refutation for goal $G_0\sigma_1$ generated by the algorithm.
\begin{align*}
G_0\textcolor{white}{\sigma_1} & \xrightarrow[B_1]{rule-1} & G_1\textcolor{white}{\sigma_1} & \xrightarrow[B_2]{rule-1} & G_2\textcolor{white}{\sigma_1} & \xrightarrow[\theta_1]{rule-2} & G_3\textcolor{white}{\sigma_1} & \xrightarrow[\theta_2]{rule-3} & G_4\textcolor{white}{\sigma_1} & \xrightarrow[B_3]{rule-1} & G_5\textcolor{white}{\sigma_1} &\xrightarrow[\theta_3]{rule-3} & (G_6\textcolor{white}{\sigma_1}=\leftarrow)\\
G_0\sigma_1 & \xrightarrow[B_1]{rule-1} & G_1\sigma_1 & \xrightarrow[B_2]{rule-1} & G_2\sigma_1 & \xrightarrow{\textcolor{white}{ab}id\textcolor{white}{-3}} & G_3\sigma_2& \xrightarrow[\epsilon]{rule-3} & G_4\sigma_3& \xrightarrow[B_3]{rule-1} & G_5\sigma_3& \xrightarrow[\epsilon]{rule-3} & (G_6\sigma_4=\leftarrow)
\end{align*}
\end{rmk}
\begin{thm}[Soundness]
Co-inductive structural resolution is sound with respect to the greatest complete Herbrand model. In other words, if a goal $G$ has a co-S-derivation with computed substitution $\sigma$, then $G\sigma$ is in the greatest model. 
\end{thm}
\begin{proof}
Given a program and some goal $G$, assume $G$ has a co-S-derivation with computed substitution $\sigma$. By Theorem~\ref{thm: co-s to co-r-id} $G\sigma$ has a co-rewriting-id derivation with computed substitution $\epsilon$, then by Proposition~\ref{thm: co-rew-id to cosld} $G\sigma$ has a co-SLD derivation with computed substitution $\epsilon$. Since co-SLD is sound with respect to the greatest complete Herbrand model, $G\sigma\epsilon=G\sigma$ is in the model.  
\end{proof}
\section{Related Work and Conclusion}\label{sec:related work}
Existing soundness proof for co-SLD helped this work. A soundness proof of co-SLD, based on co-induction, is provided in~\cite{SimonThesis2006}, in which it was established by a lemma that if some goal $G$ has co-SLD derivation with computed substitution $\sigma$, then $G\sigma$ also has a co-SLD derivation. This idea inspired the author to explore if a goal has a co-S-derivation, whether there is also a co-S-derivation for the same goal with computed answer applied. Another soundness proof of co-SLD is given in~\cite{AnconaCoSLD15}, which is based on the theory of infinite tree logic programming formulated in~\cite{JaffarInfTreeLP}. The infinite tree derivation proposed in \cite{JaffarInfTreeLP} selects all sub-goal altogether rather than one sub-goal at a time,  and it is sound with respect to the greatest model. In \cite{AnconaCoSLD15} the soundness of co-SLD is proved by showing that any co-SLD derivation can be unfolded into an infinite tree derivation, therefore the soundness of co-SLD derivation is backed by the soundness of infinite tree derivation. Our proof in this paper obviously is inspired by such technique in \cite{AnconaCoSLD15}, which relates different derivations and reuses previous results.     

As a by-product of our proof, we can use the same arguments to show that soundness and completeness of structural resolution can be proved based on soundness and completeness of SLD resolution. For soundness, if a goal $G$ has a successful structural resolution derivation with computed substitution $\sigma$, then $G\sigma$ has a successful rewriting-id derivation, which is a special case of SLD derivation. For completeness, it needs to be shown that every SLD derivation has a corresponding structural resolution derivation, by splitting each non-rewriting step in SLD derivation into one step of substitution reduction followed by one step of rewriting reduction. 

Future work will involve development of a productivity semi-decision algorithm based on co-inductive structural resolution. Now we have a sketch of the role that will be played by co-S-resolution. Given a non-terminating SLD derivation, necessarily some (maybe none) of its SLD resolution steps are rewriting reductions. A class of programs are characterized for their termination for rewriting \cite{StrResoKomendantskayaJ15,KatyaProductivityChecker16}, called \emph{observationally productive} programs. Since all consecutive rewriting steps are finite for such programs, in a non-terminating SLD derivation of some observationally productive program, there necessarily are infinite steps of non-rewriting SLD resolution steps. This fact will be crucial for productivity analysis since only non-rewriting steps can produce unifiers that may accumulate and instantiate the original goal into an infinite tree at infinity. Since the notion of productivity relies on rewriting reduction, productivity analysis is made easier by S-resolution compared with using SLD resolution; hence the advantage of structural resolution over SLD resolution. Moreover, loop detection needs to be combined with S-resolution to serve as a finite implementation of non-terminating productive S-derivations; finding a way for such a combination, proving its co-inductive soundness, and implementing it, are the contributions of this paper.             
\paragraph{Acknowledgement}
I would like to thank my supervisor Dr.~Ekaterina Komendantskaya for her support and discussion. I would like to thank Dr.~Joe Wells and anonymous reviewers for their constructive comments.

\bibliographystyle{eptcs}
\bibliography{../../yueRef}

\newpage
\appendix
\section[Implementation]{Implementation of Co-Inductive Structural Resolution}\label{app: imple}
SWI-Prolog (Multi-threaded, 64 bits, Version 7.2.3) \\ http://www.swi-prolog.org 
\begin{verbatim}
%---------------------------------------------------------
clause_tree(true,_) :- !.
clause_tree((G,R), Hypo) :-
   !,                                                     % Note 1
   clause_tree(G, Hypo),
   clause_tree(R, Hypo).

clause_tree(A,Hypo) :- find_loop(A,Hypo).

clause_tree(A, Hypo) :-    % rewriting reduction
            unifying_and_matching_rule(A, Body),
            clause_tree(Body, [A|Hypo]).                  % Note 2

clause_tree(A, Hypo) :-    % substitution reduction.
            unifying_not_matching_rule(A, _),
            clause_tree(A, Hypo).                         % Note 3

%---------------------------------------------------------
find_loop(A,[B|_]) :- A = B.
find_loop(A,[_|C]) :- find_loop(A,C).

% choose clauses whose heads unifies with the goal,
% and specifically, matches the goal.
unifying_and_matching_rule(A, Body) :-
         copy_term(A,A_copy),                             % Note 4 
         clause(A_copy,_,Ref),                            % Note 5
         clause(A1,_,Ref),                                % Note 6 
         subsumes_term(A1,A),                             % Note 7 
         clause(A,Body,Ref).                              % Note 8 

% choose clauses whose head unifies with the goal, 
% and specifically, does not match the goal.
unifying_not_matching_rule(A, Body) :-
        copy_term(A,A_copy),
        clause(A_copy,_,Ref),
        clause(A1,_,Ref),
        \+ subsumes_term(A1,A),
        clause(A,Body,Ref).
%--------------------------------------------------------
\end{verbatim}        

\begin{enumerate}[{Note} 1]
\item Clauses deal with mutually exclusive cases, hence the cuts. 

\item \verb|A| is not instantiated by finding a matching clause.  

\item \verb|A| is instantiated by finding a unifying but not matching clause. 
\item At run time variable \verb|A| is bound to the current (atomic sub-)goal $G$, \verb|A_copy| then is a variant $\acute{G}$ of $G$ with fresh variables.  Built-in  \verb|copy_term/2| is used to  make a copy  of $G$ to use in the next procedure \verb|clause(A_copy,_,Ref)| to search  for a unifying rule  without instantiating variables from $G$.

\item At run time, this procedure finds some clause whose head unifies with the variant $\acute{G}$ of the current (atomic sub-)goal $G$ and get the clause's reference number $n$ that is bound to \verb|Ref|. The term $\acute{G}$ bound to \verb|A_copy| may be instantiated. The body of the found clause, which may be instantiated, is discarded as  indicated by ``\verb|_|''.    

\item Use the reference number \verb|Ref| to get a copy of the found clause by the previous procedure `\verb|clause(A_copy,_,Ref)|'. Only the head, which is bound to \verb|A1|, of the clause is needed for subsumes check, and the body of the clause is discarded as shown by `\verb|_|'.   

\item Term matching is checked by using built-in predicate \verb|subsumes_term/2|, which does not instantiate variables. Any binding made for subsumes check will be undone by implementation of \verb|subsumes_term/2|.

\item If the subsumes check is passed by the found unifying clause, then use a fresh copy of this particular clause, as specified by \verb|Ref|, to reduce the goal. Variables in the term $t$ bound to \verb|A| will not be instantiated because the clause head subsumes $t$ as has been checked, but variables in the body of the clause, which is bound to \verb|Body|, are instantiated by sub-terms from $t$. 
\end{enumerate}

One of the anonymous reviewers of this paper suggested that 

\begin{quote}
``The two clauses for rewriting and substitution reduction can be merged into a single one to make the interpreter more compact and efficient (but maybe a bit less readable).'' 
\end{quote}
\begin{quote}
And he/she suggested the following code:
\begin{verbatim}
clause_tree(A, Hypo) :-
     copy_term(A,A_copy),
     clause(A_copy,_,Ref),
     clause(A1,_,Ref),
     subsumes_term(A1,A) *-> clause(A,Body,Ref),
                             clause_tree(Body, [A|Hypo])
                             ;
                             clause(A,Body,Ref),
                             clause_tree(A, Hypo).
\end{verbatim}
\end{quote}

\begin{verbatim}
% Example object programs
% -------------------------
% trace: clause_tree(a,[])
a  :- a1,a2.
a1 :- b1,b2.
b1 :- c1,c2.
a2.
b2.
c1.
c2.
%--------------------------
% trace: clause_tree(p(X),[]).
p(s(X)) :- p(X),p(X).  % non-linear co-recursion
%--------------------------
% trace: clause_tree(cond_f(X),[]).
cond_c(s(a)).
cond_e(s(_)).
cond_f(X) :- cond_e(X),cond_c(X).
%--------------------------
% trace: clause_tree(q(X),[]).
r(_).
p(s(X)) :- q(X).
q(X) :- p(X),r(X).

\end{verbatim}

\end{document}